\documentclass[11pt,times]{article}
\usepackage{amscd,amsfonts,amsmath,amstext,enumerate,float,amsthm,amssymb,epsfig,multirow,color,latexsym,mathrsfs,tocloft,tikz,tkz-fct,listofsymbols,tabu}
\usepackage{bbm}
\usepackage{booktabs}
\usepackage{subcaption}
\usepackage{graphicx}
\usepackage[titletoc]{appendix}
\graphicspath{ {Figures/} }
\usetikzlibrary{decorations.pathreplacing}
\usepackage[letterpaper,left=3.81cm,top=2.54cm,bottom=2.54cm,right=2.54cm]{geometry}
\usepackage[nottoc]{tocbibind}

\usepackage{appendix}
\usepackage{chngcntr}
\usepackage{etoolbox}

\usepackage{tikz,tikz-cd}
\usepackage{tkz-euclide}
\usetkzobj{all}
\usetikzlibrary{matrix,shapes,arrows}
\usetikzlibrary{positioning}
\usepackage{smartdiagram}
\AtBeginEnvironment{subappendices}{%
	\chapter*{Chapter Appendix}
	\addcontentsline{toc}{chapter}{Appendices}
	\counterwithin{figure}{section}
	\counterwithin{table}{section}
}

\usepackage{graphicx}
\usepackage{sidecap}
\usepackage{float}
\usepackage{enumerate,mdwlist}
\usepackage{epsfig}
\usepackage[english]{babel}
\usepackage{amsfonts,amsmath,enumerate}
\usepackage{amscd}
\usepackage{amsthm}
\usepackage{mleftright}
\usepackage[bbgreekl]{mathbbol}
\DeclareSymbolFontAlphabet{\mathbb}{AMSb}
\DeclareSymbolFontAlphabet{\mathbbl}{bbold}
\usepackage{colortbl}

\usepackage{subcaption}
\usepackage{amsmath,amsfonts,amssymb,mathrsfs}\usepackage{bm}
\usepackage{latexsym}
\usepackage{mleftright}
\usepackage{amssymb}
\usepackage{color}
\usepackage{xifthen}

\usepackage[numbers]{natbib}%
\usepackage[titletoc]{appendix}

\usepackage{footnote}

\newlength{\oldparindent}
\usepackage{fancyhdr}
\usepackage{hyperref} 
\hypersetup{
	colorlinks = true,
	linkcolor = blue,
	anchorcolor = blue,
	citecolor = blue,
	filecolor = blue,
	urlcolor = blue
}

\usepackage[ruled,vlined,resetcount,algosection]{algorithm2e}
\usepackage{makeidx}
\usepackage{hyperref}
\usepackage{amsmath}
\usepackage{amsthm}
\usepackage{amsfonts}
\usepackage{amssymb}
\usepackage{mathrsfs}
\usepackage{graphicx}
\usepackage{tikz,tikz-cd}
\usepackage{tkz-euclide}
\usetkzobj{all}
\usetikzlibrary{matrix}
\usepackage{smartdiagram}

\makeatletter
\newcommand{\superimpose}[2]{%
	{\ooalign{$#1\@firstoftwo#2$\cr\hfil$#1\@secondoftwo#2$\hfil\cr}}}
\makeatother

\renewcommand{\liminf}[1]{\underset{#1}{\underline{\operatorname{lim}}\,}}
\renewcommand{\limsup}[1]{\underset{#1}{\overline{\operatorname{lim}}\,}}

\newcommand{\rr}{{\mathbb{R}}}
\newcommand{\rrflex}[1]{{\ensuremath{\rr^{#1}
}}}
\newcommand{\rrD}{{\rrflex{D}}}
\newcommand{\rrd}{{\rrflex{d}}}

\newcommand{\argmin}[1]{{\ensuremath{\underset{#1}{\operatorname{argmin}}}}}
\newcommand{\arginf}[1]{\ensuremath{
		\underset{#1}{
			\operatorname{arginf}
		}
}}

\newcommand{\ee}{{\mathbb{E}}}
\newcommand{\eep}[1]{{\mathbb{E}_{\pp}\left[#1\right]}}	
\newcommand{\eecond}[2]{\ensuremath{{\ee_{\pp}\mleft[#1\middle|#2\mright]}}}
\newcommand{\eecondg}[2]{\ensuremath{{\ee_{\pp}^g\mleft[#1\middle|#2\mright]}}}
\newcommand{\eecondgp}[3]{\ensuremath{{\ee_{\pp}^{g,#3}\mleft[#1\middle|#2\mright]}}}

\newcommand{\LLp}[2]{{\mathbb{L}^{#2}_{\pp}\left(#1;\mmm\right)}}
\newcommand{\LLpgen}[3]{{\mathbb{L}^{#2}_{#3}\left(#1;\mmm\right)}}
\newcommand{\nn}{{\mathbb{N}}}

\newcommand{\pp}{{\mathbb{P}}}

\newcommand{\mmm}{{\mathscr{M}}}

\newcommand{\fff}{{\mathscr{F}}}

\newcommand{\gggg}{{\mathscr{G}}}

\newcommand{\ffft}{{\mathscr{F}_t}}

\newcommand{\gggt}{{\mathscr{G}_t}}

\newcommand{\exT}[1]{{#1}^{\mathfrak{e}:t}}
\newcommand{\exTflex}[2]{{#1}^{\mathfrak{e}:{#2}}}
\newcommand{\thetareg}[1]{{\ensuremath{
			\theta_t^{\mathfrak{ext}}
}}}
\newcommand{\thetarreg}[1]{{\ensuremath{
			\thetar{t}%
}}}

\newcommand{\EXP}[2]{{\ensuremath{\operatorname{Exp}_{#1}^g\left(#2\right)}}}
\newcommand{\LOG}[2]{{\ensuremath{\operatorname{Log}_{#1}^g\left(#2\right)}}}
\newcommand{\EXPps}[2]{{\ensuremath{\operatorname{EXP}_{#1}^g\left(#2\right)}}}
\newcommand{\LOGps}[2]{{\ensuremath{\operatorname{LOG}_{#1}^g\left(#2\right)}}}

\newcommand{\thetar}[1]{{
		\ensuremath{
			\theta_{#1}^{\mathfrak{ext}}
		}
}}

\newcommand{\GamLim}[1]{{
		\ensuremath{
			\underset{#1}{
				\operatorname{\Gamma-lim}\,	
			}
		}
}}

\makeatletter
\newcommand*\bigcdot{\mathpalette\bigcdot@{.5}}
\newcommand*\bigcdot@[2]{\mathbin{\vcenter{\hbox{\scalebox{#2}{$\m@th#1\bullet$}}}}}
\makeatother

\newcommand\independent{\protect\mathpalette{\protect\independenT}{\perp}}
\def\independenT#1#2{\mathrel{\rlap{$#1#2$}\mkern2mu{#1#2}}}
\makeatletter
\def\@chapter[#1]#2{\ifnum \c@secnumdepth >\m@ne
	\refstepcounter{chapter}%
	\typeout{\@chapapp\space\thechapter.}%
	\addcontentsline{toc}{chapter}%
	{\protect\numberline{\thechapter}\string\hypertarget{chap\thechapter}{#1}}%
	\else
	\addcontentsline{toc}{chapter}{#1}%
	\fi
	\chaptermark{#1}%
	\addtocontents{lof}{\protect\addvspace{10\p@}}%
	\addtocontents{lot}{\protect\addvspace{10\p@}}%
	\if@twocolumn
	\@topnewpage[\@makechapterhead{#2}]%
	\else
	\@makechapterhead{#2}%
	\@afterheading
	\fi}
\def\@makechapterhead#1{%
	\vspace*{50\p@}%
	{\parindent \z@ \raggedright \normalfont
		\ifnum \c@secnumdepth >\m@ne
		\huge\bfseries \@chapapp\space \thechapter
		\par\nobreak
		\vskip 20\p@
		\fi
		\interlinepenalty\@M
		\Huge \bfseries \hyperlink{chap\thechapter}{#1}\par\nobreak
		\vskip 40\p@
}}
\makeatother

\newcommand{\xireg}[1]{{\ensuremath{
			\xi_t^{\mathfrak{r}}
}}}
\newcommand{\xirreg}[1]{{\ensuremath{
			\xir{t}%
}}}

\newcommand{\xir}[1]{{
		\ensuremath{
			\xi_{#1}^{\mathfrak{r}}
		}
}}

\newtheoremstyle{dotless}{}{}{\itshape}{}{\bfseries}{}{ }{}
\theoremstyle{dotless}

\newtheorem{defn}{Definition}[section]

\newtheorem{ass}[defn]{Assumption}
\theoremstyle{plain}
\theoremstyle{definition}
\newtheorem{prop}[defn]{Proposition}%
\newtheorem{cor}[defn]{Corollary}%
\newtheorem{lem}[defn]{Lemma}%
\newtheorem{ex}[defn]{Example}%
\newtheorem{thrm}[defn]{Theorem}%
\newtheorem{remark}[defn]{Remark}%
\author{Anastasis Kratsios}
\author{Cody B. Hyndman}

\title{\Large \bf  Non-Euclidean Conditional Expectation and Filtering} %
\author{\textsc{Anastasis Kratsios}\thanks{Department of Mathematics and Statistics, Concordia University, 1455 Boulevard de Maisonneuve Ouest, Montr\'{e}al, Qu\'{e}bec, Canada H3G 1M8. %
		emails: anastasis.kratsios@mail.concordia.ca, cody.hyndman@concordia.ca
	} \quad
	\textsc{Cody B.\ Hyndman\footnotemark[1] 
	}\\
}
\date{%
	\today \\ %
}

\begin{document}
	\newpage
\maketitle
\numberwithin{equation}{section} 

\begin{abstract}
A non-Euclidean generalization of conditional expectation is introduced and characterized as the minimizer of  expected intrinsic squared-distance from a manifold-valued target.  The computational tractable formulation expresses the non-convex optimization problem as transformations of Euclidean conditional expectation.  This gives computationally tractable filtering equations for the dynamics of the intrinsic conditional expectation of a manifold-valued signal and is used to obtain accurate numerical forecasts of efficient portfolios by incorporating their geometric structure into the estimates.  
\end{abstract}

\noindent
{\itshape Keywords:} Conditional Expectation, Non-Euclidean Conditional Expectation, Portfolio Theory, Non-Linear Conditional Expectation, Non-Euclidean Geometry, Stochastic Filtering, Non-Euclidean Filtering Equations, Gamma-Convergence.  

\noindent
{\let\thefootnote\relax\footnotetext{This research was supported by the Natural Sciences and Engineering Research Council of Canada (NSERC).  The authors thank Alina Stancu (Concordia University) for helpful discussions.}}

\noindent
{\bf Mathematics Subject Classification (2010):} 60D05, 91G60, 91G10, 62M20, 60G35.

\section{Introduction}
Non-Euclidean geometry occurs naturally in problems in finance.  
Short-rate models consistent with finite-dimensional Heath-Jarrow-Morton (HJM) models are characterized using Lie group methods, in \cite{bjork2010interest}. In \cite{henry2005general,henry2008analysis}, highly accurate stochastic volatility model estimation methods are derived using Riemannian heat-kernel expansions.  In \cite{filipovic2001consistency}, the equivalent local martingale measures (ELMMs) of finite-dimensional term-structure models for zero-coupon bonds are characterized using the smooth manifold structure associated with factor models for the forward-rate curve.  In \cite{brody2004chaos}, information-geometric techniques for yield-curve modeling which
consider finite-dimensional manifolds of probability densities are developed.  %
In \cite{han2015geometric,han2016geometric}, Riemannian geometric approaches to stochastic volatility models and covariance matrix prediction are employed to successfully predicts stock prices.  In \cite{kratsios2017geometrics}, it is shown that considering a relevant geometric structures on a mathematical finance problem leads to more accurate out-of-sample forecasts.  The superior forecasting power of non-Euclidean methods is interpreted as encoding information present in mathematical finance problems which is otherwise overlooked by the classical Euclidean methods.  Each of these methodologies approach distinct problems in mathematical finance using differential geometry.

Conditional expectation and stochastic filtering are some of the most fundamental tools used in applied probability and finance.  Geometric formulations of conditional expectation, such as those used in \cite{snoussi2006particle,LeJanwAtanabeFilteringSalemSIAM} are solutions to non-convex optimization problems.  The non-convexity of the problem makes computation of these formulations of non-Euclidean conditional expectations difficult or intractable. 

Non-Euclidean filtering formulations such as those of \cite{FilteringDuncan}, \cite{ng1984nonlinear}, or \cite{Elworthygeofilt} assume that the signal and/or noise processes are non-Euclidean and estimate functionals of the noisy signal using the classical Euclidean conditional expectation.  In \cite{LeJanwAtanabeFilteringSalemSIAM} dynamics for the intrinsic conditional expectation of a manifold-valued signal was found, using the Le Jan-Watanabe connection.  This connection reduced the intrinsic non-Euclidean filtering problem to a Euclidean filtering problem.  However, the authors of \cite{LeJanwAtanabeFilteringSalemSIAM} remark that implementing their results may be intractable due to the added complexity introduced by the Le Jan-Watanabe connection.  

This paper presents an alternative computationally tractable characterization of intrinsic conditional expectation, called geodesic conditional expectation, and uses it to produce a computable solution to a non-Euclidean filtering problem similar to that of \cite{LeJanwAtanabeFilteringSalemSIAM}.  The implementation is similar to \cite{snoussi2006particle} for a non-Euclidean particle filter.  However, in \cite{snoussi2006particle} the convergence of the algorithm to the non-Euclidean conditional expectation is left unjustified.  The geodesic conditional expectation expresses the intrinsic conditional expectation as a limit of certain transformations of Euclidean conditional expectations associated to the non-Euclidean signal process.  Analogous to \cite{LeJanwAtanabeFilteringSalemSIAM}, these transformations reduce the computation of the non-Euclidean problem to the computation of a Euclidean problem, with the central difference being that the required transformations are available in closed form.  The infinitesimal linearization transformations considered here are similar to those empirically postulated in the engineering, computer-vision, and control literature in \cite{fletcher2004principal,fletcher2011geodesic,han2015geometric,hauberg2013unscented,allerhand2011robust,snoussi2006particle}.

The paper is organized as follows.  Section $2$ introduces the necessary notation and the general terminology through the paper.  In Section $3$, the relationship between portfolio selection and non-Euclidean geometry is introduced, and elements of Riemannian geometry are reviewed through the lens of the space of efficient portfolios.  Section $4$ introduces two natural generalizations of conditional expectation to the non-Euclidean setting.  Both formulations are shown to be equivalent in Theorem~\ref{thrm_main}.  Corollary~\ref{cor_Main_2} provides non-Euclidean filtering equations which describe the dynamics of the non-Euclidean expectations.  Using Corollary~\ref{cor_Main_2}, Section $5$ returns to the space of efficient portfolios and numerically illustrates how efficient portfolios, on historical stock data, can be more precisely forecasted by incorporating geometric features into the estimation procedure.  Our filtering procedure is benchmarked against other intrinsic filtering algorithms from the engineering and computer vision literature.  
Section $6$ reviews the contributions of this paper.  

\section{Preliminaries and Notation}
In this paper $(\Omega,\fff,\fff_{\bigcdot},\pp)$ denotes a complete stochastic base on which independent Brownian motions, denoted by $W_t$ and $B_t$ are defined.  Furthermore, $\gggg_{\bigcdot}$ will denote a sub-filtration of $\fff_{\bigcdot}$.  The vector-valued conditional expectation will be denoted by $\eecond{X_t}{\gggg}$.  

The measure $m$ %
will denote the Lebesgue %
measure, $L^p_{m}(\fff;\rrD)$ %
will denote the Bochner-Lebesgue spaces for $\fff$-measurable $\rrD$-valued functions with respect to the $D$-tuples of Lebesgue measure $m$%
.  If $D=1$, the Bochner-Lebesgue spaces will be abbreviated by $L^p_{m}(\fff)$%
.  For a Riemannian manifold $(\mmm,g)$, the intrinsic measure is denoted by $\mu_g$ and the induced distance function is denoted by $d_g$.  %
The disjoint union, or coproduct, of topological spaces will be denoted by $\coprod$.  The set of c\`{a}dl\`{a}g paths from $\rr$ into the metric space induced by $(\mmm,g)$, is defined by $D(\rr;\mmm,d_g)$.  

The next section motivates the geometries studied in this paper by introducing and discussing the geometry of efficient portfolios.  
\section{The Geometry of Efficient Portfolios}\label{ss_GEOEFFPORT}
A fundamental problems in mathematical finance is choosing an optimal portfolio.  %
Typically, in modern portfolio theory, a portfolio is comprised of $D$ predetermined risky assets and a riskless asset.  Efficient portfolios are portfolios having the greatest return but not exceeding a fixed level of risk.  Classically, the return level is measured by the portfolio's expected (log)-returns.  The portfolio's risk is quantified as the portfolio's variance.  The optimization problem defining efficient portfolios may be defined in a number of ways, the one considered in this paper is the following Sharpe-type ratio
\begin{equation}
\hat{w}(\gamma,\mu,\Sigma)\triangleq 
\underset{
	\underset{
		\bar{1}^{\star}w=1
	}{w \in \rr^D}
}{
	\operatorname{argmin}}
\left(
-\gamma\mu^{\star}w +\frac{
	w^{\star}\Sigma w
}{2}
\right)
.
\label{eq_defn_Sharp_ratio}
\end{equation}
Here $w$ is the vector of portfolio weights expressed as the proportion of wealth invested in each risky asset, $\mu \in \rr$ is the vector of the expected log-returns of the risky assets, $\Sigma$ is the covariance matrix of those log-returns, $\gamma$ is a parameter balancing the objectives of maximizing the portfolio return versus minimizing the portfolio variance, $\bar{1}$ is the vector with all its components equal to $1$, and $^\star$ indicates matrix transpose operation.  
If $\Sigma$ is not degenerate, the unique optimal solution to equation~\eqref{eq_defn_Sharp_ratio} is
\begin{equation}
\hat{w}(\gamma,\mu,\Sigma) = 
\frac{\Sigma^{-1}\bar{1}}{\bar{1}^{\star} \Sigma \bar{1}}
+ \gamma \left( \Sigma^{-1}\mu - \frac{\bar{1}^{\star}\Sigma^{-1}\mu}{\bar{1}^{\star}\Sigma^{-1}\bar{1}}
\Sigma^{-1}\bar{1}
\right)
.
\label{eq_optim_portfolio}
\end{equation}
The particular case where $t$ is set to $0$ is the minimum variance portfolio of \cite{markowitz1968portfolio}.  The minimum-variance portfolio $\hat{w}(0,\mu,\Sigma)$ may also be derived by minimizing the portfolio variance subject to the budget constraint $\bar{1}^\star w=1$.  By adding a risk-free asset to the portfolio, one can derive similar expressions for the market portfolio and the capital market line (for more details on this approach to portfolio theory see \cite{BestOptimizationPortfolio}).  %

Unlike the returns vector $\mu$, a portfolio's covariance matrix is not meaningfully represented in Euclidean space.  That is, a covariance matrix does not scale linearly and the difference of covariance matrices need not be a covariance matrix.  
Therefore, forecasting a future covariance matrix, even through a simple technique such as linear regression directly to the components of $\Sigma$, can lead to meaningless forecasts.  Using the intrinsic geometry of the set of positive-definite matrices, denoted by $\mathscr{P}_D^+$, avoids these issues.  %

The space $\mathscr{P}_D^+$, has a well studied and rich geometry lying at the junction of Cartan-Hadamard geometry and Lie theory.  Empirical exploitation of this geometry has found many applications in mathematical imaging (see \cite{moakher2011riemannian}), computer vision (see \cite{pennec2006riemannian}), and signal processing (see \cite{barbaresco2008innovative}).  Moreover, connections between this geometry and information theory have been explored in \cite{smith2005covariance}, linking it to the Cramer-Rao lower bound.  %

The set $\mathscr{P}_D^+$ is smooth and comes equipped with a natural infinitesimal notion of distance called Riemannian metric.  Denoted by $g$, the Riemannian metric on $\mathscr{P}_D^+$ quantifies the difference in making infinitesimal movements in Euclidean space along $\mathscr{P}_D^+$ to making infinitesimal movements with respect to the geometry of $\mathscr{P}_D^+$.  The description of Riemannian manifolds as subsets of Euclidean space is made rigorous by Nash in the embedding theorem in \cite{nash1956imbedding}.  Distance between two points on $\mathscr{P}_D^+$ is quantified by the length of the shortest path connecting the two points, called a geodesic.  On $\mathscr{P}_D^+$, any two points can always be joined by geodesic.  The distance function taking two points to the length of the unique most efficient curve joining them can be expressed as
\begin{equation}
d_g^2\left(
\Sigma_1
,
\Sigma_2
\right)
\triangleq 
\left\|
\log\left(
\Sigma_2^{\frac{1}{2}}\Sigma_1\Sigma_2^{\frac{1}{2}}
\right)
\right\|_F^2
= \sum_{i=1}^{d}
\lambda_i^2\left(
\log\left(
\Sigma_2^{-\frac{1}{2}}\Sigma_1\Sigma_2^{-\frac{1}{2}}
\right)
\right)
.
\label{eq_dist_mani_PD}
\end{equation}
The function $d_g$ makes $\mathscr{P}_D^+$ into a complete metric space, where the distance between two points corresponds exactly to the length of the unique distance minimizing geodesic connecting them.  Where, $\|\cdot\|_F$ is the Frobenius norm, which first treats a matrix as a vector and subsequently computes its Euclidean norm, ${\Sigma}^{\frac1{2}}$ is the matrix square-root operator, $\log$ is the matrix logarithm, and $\lambda_i(\Sigma)$ denotes $i^{th}$ eigenvalue of $\Sigma$.  Both the log and ${\Sigma}^{\frac1{2}}$ operators are well-defined on $\mathscr{P}_D^+$.

The disparity between the distance measurements is explained by the intrinsic curvature of $\mathscr{P}_D^+$.  Sectional curvature is a formalism for describing curvature intrinsically to a space, such as $\mathscr{P}_D^+$.  It is measured by sliding a plane tangentially to geodesic paths and measuring the twisting and turning undergone by that tangential plane.  %
A detailed measurement of $\mathscr{P}_D^+$ shows that its sectional curvature is everywhere non-positive.  This means that locally the space $\mathscr{P}_D^+$ is locally curved somewhat between a pseudo-sphere and Euclidean space.  Alternatively this can be described by stating that $\mathscr{P}_D^+$ nowhere bulges out like a circle but is instead puckered in or flat.  

A smooth subspace of Euclidean space having everywhere non-positive curvature when equipped with a Riemannian metric, and for which every pair of points can be joined by a unique distance minimizing geodesic is called a Cartan-Hadamard manifold.  These spaces posses many well-behaved properties, as studied in \cite{NonposCurvBall}, but for this discussion the most relevant property of Cartan-Hadamard manifolds to this paper is the existence of a smooth map $\LOG{}{}$ from $\mathscr{P}_D^+\times \mathscr{P}_D^+$ onto $\rrd$.  Here $\rrd$ is the Euclidean space of equal dimension to $\mathscr{P}_D^+$.  For every fixed input, this map is infinitely differentiable, has an infinitely differentiable inverse and therefore puts $\mathscr{P}_D^+$ in smooth correspondence with $\rrd$.  The map $\LOG{}{}$ is called the Riemannian Logarithm.  It is related to the distance between two covariance matrices through
\begin{align}
\nonumber
d_g\left(\Sigma_1,\Sigma_2\right) =& \|\LOG{\Sigma_1}{\Sigma_2}\|_2,
\\
Log_{\Sigma_1}(\Sigma_2) \triangleq &
\Sigma_1^{\frac{1}{2}}log\left(
\Sigma_1^{-\frac{1}{2}}\Sigma_2\Sigma_1^{-\frac{1}{2}}
\right)
\Sigma_1^{\frac{1}{2}}
.
\label{eq_log_PD_pluss}
\end{align}
The Riemannian Exponential map, denoted by $\EXP{}{}$, is the inverse of $\LOG{}{}$.  The Riemannian Exponential map takes a covariance matrix $\Sigma_1$ and a tangential velocity vector $v$ to $\Sigma_1$, and maps it to the covariance matrix $\Sigma_2$, found by traveling along $\mathscr{P}_D^+$ at the most efficient path beginning at $\Sigma_1$ with initial velocity $v$ and stopping the movement after one time unit.  Geodesics on $\mathscr{P}_D^+$ are obtained by scaling the initial velocity vector in the $\EXP{}{}$ map, which is expressed as
\begin{align}
\label{eq_exp_PD_pluss}
\EXP{\Sigma_1}{v} \triangleq &
\Sigma_1^{\frac{1}{2}}exp\left(
\Sigma_1^{-\frac{1}{2}}Sym(v)\Sigma_1^{-\frac{1}{2}}
\right)
\Sigma_1^{\frac{1}{2}}\\
\nonumber
Sym(v)\triangleq &
\begin{pmatrix}
v_{1}  & v_{2} &  \ldots  & v_{D} \\
v_{2} & v_{D+1} &  \ldots  & v_{2D-1} \\
\vdots&  & \ddots & \vdots \\
v_D &  & \ldots      & v_{\frac{D(D+1)}{2}}.
\end{pmatrix} 
,
\end{align}
where $exp$ is the matrix exponential.

Returning to portfolio theory, any efficient portfolio in the sense of equation~\eqref{eq_optim_portfolio}, is entirely characterized by the log-returns, the non-degenerate covariance structure between the risky assets, and the risk-aversion level.  The space parameterizing all the efficient portfolios, which will be called the Markowitz space after \cite{markowitz1968portfolio}, has a natural geometric structure.
\begin{defn}[Markowitz Space]\label{defn_Mark}
	Let $g_E^1,g_E^D,$ and $g$ be the Euclidean Riemannian metrics on $\rr$, $\rrD$, and the Riemannian metric on $\mathscr{P}_D^+$.  The Riemannian manifold $$\left(\mathscr{M}^{Mrk}_D,g^{Mrk}_D\right)\triangleq \left(
	\rr\times \rrd\times \mathscr{P}_D^+,g_E^1\oplus g_E^2\oplus g
	\right)$$ is called the ($D$-dimensional) Markowitz space.  
\end{defn}
\begin{prop}[Select Properties of the Markowitz Space]\label{prop_Mark_space_properties}
	The Markowitz space is connected, of non-positive curvature, and its associated metric space is complete.  The distance function is 
	\footnotesize\begin{align}
	d_{Mrk}^2\left((\gamma_1,\mu_1,\Sigma_1),
	(\gamma_2,\mu_2,\Sigma_2)
	\right) \triangleq &
	\|\gamma_2-\gamma_1\|_2^2 + \|\mu_2-\mu_1\|_2^2
	+ 
	\left(\sum_{i=1}^{d}
	\lambda_i^2\left(
	\log\left(
	\Sigma_2^{-\frac{1}{2}}\Sigma_1\Sigma_2^{-\frac{1}{2}}
	\right)
	\right)\right)^2
	.
	\label{prop_MARK_space_properties_metric}
	\end{align}\normalsize
	The Riemannian $\LOG{}{}$ and $\EXP{}{}$ maps on $\mmm^{Mrk}$ are of the form
	\begin{align}
	\nonumber
	\EXP{(\gamma_1,\mu_1,\Sigma_1)}{(v_1,v_2,v_3)} \triangleq &
	\left(
	\gamma_1 + v_1,\mu_1 + v_2,
	\Sigma_1^{\frac{1}{2}}exp\left(
	\Sigma_1^{-\frac{1}{2}}Sym(v_3)\Sigma_1^{-\frac{1}{2}}
	\right)
	\Sigma_1^{\frac{1}{2}}
	\right) \\
	\LOG{(\gamma_1,\mu_1,\Sigma_1)}{
		(\gamma_2,\mu_2,\Sigma_2)
	} \triangleq &
	\left(
	\gamma_2-\gamma_1,\mu_2-\mu_1,
	\Sigma_1^{\frac{1}{2}}log\left(
	\Sigma_1^{-\frac{1}{2}}\Sigma_2\Sigma_1^{-\frac{1}{2}}
	\right)
	\Sigma_1^{\frac{1}{2}}
	\right)
	\label{prop_MARK_space_properties_EXPLOG}
	.
	\end{align}
	Note that the Riemannian exponential and logarithm maps are defined everywhere and put $\mmm^{Mrk}$ in a smooth $1$ to $1$ correspondence with $\rrflex{1+D+\frac{D(D+1)}{2}}$.  
\end{prop}
\begin{proof}
	The proof is deferred to the appendix.  
\end{proof}
The Markowitz space serves as the prototypical example of the geometric spaces considered in the rest of this paper, these are Riemannian manifolds, of non-positive curvature, for which every two points can be joined by a unique distance minimizing geodesic.  In the remainder of this paper, all Riemannian manifolds will be Cartan-Hadamard manifolds.  
Cartan-Hadamard manifolds appear in many places in mathematical finance, for example in \cite{henry2008analysis} the natural geometry associated with stochastic volatility models with two driving factors are Cartan-Hadamard manifolds.  

On Cartan-Hadamard manifolds, such as the Markowitz space, there is no rigorously defined notion of conditional expectation.  Therefore rigorous estimation intrinsic to these spaces' geometries is still a generally unsolved problem.  We motivate this problem by discussing a few formulations of intrinsic conditional expectation and related empirical techniques present in the mathematical imaging literature.

The least-squares formulation of conditional expectation is 
$$
\eecond{X_t}{\gggg}\triangleq\operatorname{argmin}_{Z\in L^2_{\pp}(\gggg;\rrd)} \eep{\|X_t - Z\|_2^2}
.
$$
Replacing the expected Euclidean distance by the expected intrinsic distance gives the typical formulation of a non-Euclidean conditional expectation.  This formulation will be referred to as intrinsic conditional expectation.  

Alternatively, estimates in a Riemannian manifold are made by locally linearizing the data using the Riemannian log map, performing the estimate in Euclidean space, and returning the data back onto the manifold.  This type of methodology has been used extensively in the computer vision and mathematical imaging literature by \cite{fletcher2004principal,han2015geometric,hauberg2013unscented,allerhand2011robust}, and \cite{snoussi2006particle}.  In \cite{snoussi2006particle}, the authors empirically support estimating the intrinsic conditional expectation a following procedure which first linearizes the observation using the Riemannian Log transform, subsequently computes the conditional expectation in Euclidean space, and lastly returns the prediction onto the Riemannian manifold using the Riemannian Exp map.  

This paper provides a rigorous framework for the two methods described above, proves the existence of their optimum, and shows that both formulations agree.  The rigorous formulation of the non-Euclidean filtering algorithm of \cite{snoussi2006particle} is used to derive non-Euclidean filtering equations.  The non-Euclidean filtering problem is implemented and used to accurately forecast efficient portfolios by exploiting the geometry of the Markowitz space.  

Empirical evidence for the importance of considering non-Euclidean geometry will be examined in the next section before developing a general theory of non-Euclidean conditional expectation in Section~\ref{sec_EstLOB}.  

\section{Non-Euclidean Conditional Expectations and Intrinsic Forecasting}\label{sec_EstLOB}

Let $\eecond{X_t}{\gggt}$ denote the vector-valued conditional expectation in $\rrd$.  Let $0<\Delta <t$ and consider %
\begin{align}
\eecond{X_t}{\gggt} 
= &
\lim\limits_{\Delta \downarrow 0}\eecond{X_t}{\gggt} \nonumber \\ 
= & 
\left(
\lim\limits_{\Delta \downarrow 0}\eecond{X_{t-\Delta}}{\gggg_{t -\Delta}} + 
\eecond{\left(X_t - \eecond{X_{t-\Delta}}{\gggg_{t -\Delta}}\right)}{\gggt}
\right) \nonumber
\\
=& %
\lim\limits_{\Delta \downarrow 0}
\EXP{
	\eecond{X_{t-\Delta}}{\gggg_{t-\Delta}}
}{
	\eecond{
		\LOG{
			\eecond{X_{t-\Delta}}{\gggg_{t-\Delta}}
		}{
			X_t
		}
	}{\gggt}
}
.
\label{eq_computor}
\end{align}
The first equality is obtained by taking the limit of a constant sequence and the second line it achieved using the $\gggt$-measurability of $\eecond{X_{t-\Delta}}{\gggg_{t-\Delta}}$ and the linearity of conditional expectation.  
The last line of equation~(\ref{eq_computor}) is obtained by using the fact that the Riemannian Exponential and Logarithm maps in Euclidean space respectively correspond to addition and subtraction.  

Equation~\eqref{eq_computor} expresses the conditional expectation at time $t$ as moving from the conditional expectation at an arbitrarily close past time along a straight line with initial velocity given determined by the position of $X_t$ and the last computed conditional expectation.  The past time-period is made arbitrarily small by taking the limit  $\Delta \rightarrow 0$.  

Equation~\eqref{eq_computor} may be generalized and taken to be the definition of conditional expectation in the general Cartan-Hadamard manifold setting.   In general, this definition will rely on a particular non-anticipative pathwise extension of a process.  The definition of this pathwise extension is similar to the horizontal path extensions introduced in \cite{dupire1994pricing}. The extension $\exT{X}_t$ of a process $X_t$ holds the initial realized value $X_0$ constant back to $-\infty$ and the time $t$ value constant all the way to $\infty$.  Formally, $\exT{X}_t$ is defined pathwise by
\begin{figure}[H]
	\centering
	\begin{equation*}
	\exT{X}_s(\omega)\triangleq \begin{cases}
	X_t(\omega) & t\leq s\\
	X_s(\omega) & 0\leq s\leq t\\
	X_0(\omega) & s\leq 0
	\end{cases}
	.
	\end{equation*}
	\caption{Extension of the process $X_t$. }
	\label{eq_extension_back}
\end{figure}
The next assumption will be made to ensure that the initial conditional probability laws exist on $\mmm$.  
\begin{ass}\label{lem_IC_reg}
	Suppose that $X_0$ is $\gggg_0$-measurable and is absolutely continuous with respect to the intrinsic measure $\mu_g$ on $(\mmm,g)$.  Denote its density by $f_0$, and assume that there exists at-least one point $x_0$ in $\mmm$ such that the integral $\int_{y \in \mmm}
	d_g^2\left(x_0,y\right)
	f_0(y)
	\mu_g(dy)
	$ is finite.  
\end{ass}
\begin{defn}[Geodesic Conditional Expectation]\label{defn_ICEnew}
	Let $X_t$ be an $(\mmm,g)$-valued c\`{a}dl\`{a}g process and $\gggg_{t}$ be a sub-filtration of $\fff_{t}$.  The geodesic conditional expectation of $X_{t}$ given $\gggg_{t}$, denoted by $X_t^{g}$ is defined to be the solution to the recursive system
	\begin{align}
	\label{eq_defn_ICEnew_1}
	X_{t}^{g}
	\triangleq &
	\begin{cases}
	\lim\limits_{n \mapsto \infty}
	\EXP{
		X_{t -\frac1{n}}^{g}
	}{
		\eecond{
			\LOG{
				X_{t -\frac1{n}}^{g}
			}{
				\exTflex{X}{t}_t
			}
		}{\gggt}^o
	} & \mbox{ if} t>0\\
	\argmin{x \in \mmm} 
	\int_{y \in \mmm}
	d_g^2\left(x,y\right)
	f_0(y)
	\mu_g(dy) & \mbox{ if } t \leq 0
	,
	\end{cases}
	\end{align}
	where $Y^{o}$ is the $\gggt$-optional projection.  
\end{defn}
The geometric intuition behind equation~\eqref{eq_defn_ICEnew_1} is that the current  geodesic conditional expectation at time $t$ is computed by first predicting the infinitesimal velocity describing the current state on $(\mmm,g)$ from the previous estimate at time $t-\frac1{n}$, and then moving across the infinitesimal geodesic along $(\mmm,g)$ in that direction.  The computational implication of equation~\eqref{eq_defn_ICEnew_1} is that all the classical tools for computing the classical Euclidean conditional expectation may be used to compute the geodesic conditional expectation, once the Riemannian Exp and Riemannian Log maps are computed.  
\begin{lem}[Existence of Initial Condition]\label{lem_IC_Comp}
	Under Assumption~\ref{lem_IC_reg},	
	$X_0^g$ exists and is $\pp$-a.s\ unique.  
\end{lem}
\begin{proof}
	Under Assumption~\ref{lem_IC_reg}, \citep[Exercise 5.11]{NonposCurvBall} guarantees the existence of %
	$X_0$.    
\end{proof}

Geodesic conditional expectation is an atypical formulation of non-Euclidean conditional expectation.  Typically, non-Euclidean conditional expectation is defined as the $\mmm$-valued random element minimizing the expected intrinsic distance to $X_t$.  

Following \cite{SobolevRiemannian}, by first isometrically embedding $(\mmm,g)$ into a large Euclidean space $\rrD$, the space $L^p_{\pp}\left(\fff;\mmm\right)$ is subsequently defined as the subset of the Bochner-Lebesgue space $L^p_{\pp}\left(\fff;\rrD\right)$ consisting of the equivalence classes of measurable maps
which are $\pp$-a.s. supported on $\mmm$, and for which there exists some $\hat{X} \in \mmm$ for which
\begin{equation}
\left({\int_{\omega \in \mmm} d_g^p\left(X(\omega),\hat{X}\right)\pp\left(d\omega\right)}\right)^{\frac1{p}}
<\infty
\label{eq_defn_LebesgueIntrinsicFinite}
.
\end{equation}
The set $L^p_{\pp}\left(\fff;\mmm\right)$ is a Banach manifold (see \cite{tausk2002banach} for more general results).  
\begin{defn}[Intrinsic Conditional Expectation]\label{defn_ICEold}
	The intrinsic conditional expectation with respect to the $\sigma$-subalgebra $\gggt$ of $\fff$ of an $\mmm$-valued stochastic process $X_t$ is defined as the optimal Bayesian action
	$$
	\eecondgp{X_t}{\gggt}{p} \triangleq \arginf{Z_t \in L^p_{\pp}\left(\gggt;\mmm\right)}
	\eep{d_g^p(Z_t,X_t)}.
	$$
	When $p=2$, we will simply write $\eecondg{X_t}{\gggt}$.  
\end{defn}

Intuition about intrinsic conditional expectation is gained by turning to the Markowitz space.  
\begin{ex}%
	\label{ex_Mrk_cond}
	Let $\gamma\geq 0$ be fixed and constant.  Let $X_t\triangleq(\gamma,\mu_t,\Sigma_t)$ be a process taking values in the Markowitz space, equation~\eqref{prop_MARK_space_properties_metric}.  Then the intrinsic conditional expectation of $X_t$ given $\gggt$ is 
	\small\begin{equation}
	\begin{aligned}
	\eecondg{(\gamma,\mu_t,\Sigma_t)}{\gggt} = &
	\argmin{(\tilde{\mu}_t,\tilde{\Sigma}_t) \in 
		\LLp{\gggt}{p}
	}
	\eep{\|\mu_t-\tilde{\mu}_t\|_2^2} \\
	& + \eep{
		\left(	\sum_{i=1}^{d}
		\lambda_i^2\left(
		\log\left(
		\sqrt{\tilde{\Sigma}_t}^{-1}\Sigma_t\sqrt{\tilde{\Sigma}_t}^{-1}
		\right)
		\right)\right)^2
	}
	.
	\end{aligned}
	\label{eq_Mark_reg}
	\end{equation}\normalsize
	The conditional expectation intrinsic to the Markowitz space seeks portfolio weights which give the most likely log-returns given the information in $\gggt$, while penalizing for the variance taken on by following that path. 
	
	In the case where $\Sigma_t$ is independent of $\mu_t$ and $\Sigma_t$ is $\gggt$-measurable, equation~\eqref{eq_Mark_reg} simplifies.  Since $\mu_t$ does not depend on $\Sigma_t$ and the latter is in $\LLp{\gggt}{2}$, $\Sigma_t$ may be substituted into the second term, which sets it to zero.  Therefore, in this simplified scenario the least-squares property of Euclidean conditional expectation (see \citep[Page 80]{kallenberg2006foundations}) that
	\begin{equation}
	\eecondg{(\gamma,\mu_t,\Sigma_t)}{\gggt} =
	\argmin{(\tilde{\mu}_t,\tilde{\Sigma}_t) \in 
		\LLp{\gggt}{p}
	}
	\eep{\|\mu_t-\tilde{\mu}_t\|_2^2}
	= \eecond{\mu_t}{\gggt}
	.
	\label{eq_Mark_reg_reduce}
	\end{equation}
\end{ex}
There is a natural topology defined on $L^p_{\pp}(\gggt;\mmm)$ which is characterized as being the weakest topology on which sequences of c\'{a}dl'{a}g process process $\left\{X_{t-\frac{1}{n}}\right\}_{n \in \nn}$ in $L^p_{\pp}(\gggt;\mmm)$ converge to $X_t$ in $L^p_{\pp}(\gggt;\mmm)$ (see~\ref{ref_rel_top} for a rigorous discussion).  For any two elements $X$ and $Y$ of $L^p_{\pp}(\gggt;\mmm)$ with this topology, we will write
$$
X\equiv
Y,
$$
if $X$ and $Y$ are indistinguishable in this topology.  Intuitively, this means that they cannot be further separated in the topology.  For example in $\rrD$ two points are indistinguishable if and only if they are equal, the same is true for example in metric spaces.  Whereas in the space of measurable functions from $\rr$ to itself which are square integrable equipped with its usual topology, two functions are inst indistinguishable if and only if they are equal on almost all points (see \cite{Kelleytop} for details on topological indistinguishability.)

Under mild assumptions, the geodesic conditional expectation and intrinsic conditional expectation agree on Cartan-Hadamard spaces as shown in the following theorem.
\begin{thrm}[Unified Conditional Expectations]\label{thrm_main}
	Let $X_t$ be an $\mmm$-valued process with c\`{a}dl\`{a}g paths which is in $L^2(\gggt;\mmm)$ for $m$-a.e. $t\geq 0$ and is such that Assumptions~\ref{lem_IC_reg} and~\ref{ass_finite} hold.  
	For $1\leq p<\infty$, the intrinsic conditional expectation $\eecondg{X_t}{\gggt}$ exists. Moreover, if $p=2$, then 
	\begin{align}
	\label{eq_thrm_main1}
	\exTflex{\eecondg{\exT{X}_t}{\gggt}}{t}
	\equiv
	\exT{X^g_t}
	,
	\end{align}
	where the left-hand side of equation~\eqref{eq_thrm_main1} is the intrinsic conditional expectation and its right-hand side is the geodesic conditional expectation.  
\end{thrm}
Theorem~\ref{thrm_main} justifies the particle filtering algorithm of \cite{snoussi2006particle}.  Before proving Theorem~\ref{thrm_main} and developing the required theory, a few implications and examples will be explored.  
\subsection{Filtering Equations}
Theorem~\ref{thrm_main} has computational implications in terms of forecasting the optimal intrinsic conditional expectation using the geodesic conditional expectation.  These implications are in the computable solution to the certain filtering problems.  

Instead of discussing the dynamics of a coupled pair of $\mmm$-valued signal process $X_t$ and observation processes $Y_t$ intrinsically to $(\mmm,g)$, Theorem~\ref{thrm_main} justifies locally linearizing $X_t$ and $Y_t$, then subsequently describing their Euclidean dynamics before finally returning them onto $\mmm$. %
More, specifically assume that
\begin{equation}
\begin{aligned}
\tilde{X}_t^i  = & \int_0^tf^i(\tilde{X}_u^i)du + \int_0^t\beta^i(u,\tilde{X}_u^i)dB_u^i,\\
\tilde{Y}_t^i = & 
\int_0^u c^i(\tilde{X}_u^i,\tilde{Y}_u^i)du + \int_0^t\alpha^i(u,\tilde{Y}_u^i) dW_u^i,\\
\tilde{X}_t^i \triangleq &\langle \LOG{X_{t-\frac1{n}}^g}{X_t},e_i\rangle_{\rrd} \\
\tilde{Y}_t^i \triangleq & 
\langle \LOG{X_{t-\frac1{n}}^g}{Y_t},e_i\rangle_{\rrd} \\
B^i\independent W^j, & \, B^i\independent B^j, \, W^i\independent W^j ; i\neq j
\end{aligned}
\label{eq_coupling}
\end{equation}
where $B_t$ and $W_t$ are independent Brownian motions and $\tilde{X}_t^i,\tilde{Y}_t^i$ satisfy the usual existence and uniqueness conditions (see \citep[Chapter 22.1]{CohenAndElliotStochAppl} for example).  This implies that $X_t^i$ depends on only itself and that $\tilde{Y}_t^i$ depends only one $X_t^i$ and itself.  In particular, this implies that
\begin{equation}
\eecond{\tilde{X}_t^i}{\gggt} =\eecond{\tilde{X}_t^i}{\gggt^i}
\label{eq_indep_simplifications}
,
\end{equation}
where $\gggg_t^i$ is the filtration generated only by $\tilde{Y}_t^i$.  Using these dynamics, asymptotic local filtering equations for the dynamics of the geodesic conditional expectation $X_t^{g}\triangleq \eecondg{X_t^{\mathfrak{ext}}}{\gggt}$ in terms of $\eta_t^{\mathfrak{ext}}$ can be deduced and are summarized in the following Corollary of Theorem~\ref{thrm_main}.  
\begin{cor}[Asymptotic Non-Euclidean Filter]\label{cor_Main_2}
	Let $\mmm=\rrd$, denote the $i^{th}$ coordinate of $X_t^{g}$ by $X_t^{g,i}$, 
	and suppose Assumptions~\ref{lem_IC_reg} and~\ref{ass_finite} as well as the assumptions on $X_t$ and $Y_t$ made in \citep[Chapter 22.1]{CohenAndElliotStochAppl}.  If $X_t^g$ is $\pp\otimes m$-a.e. unique, a version of the intrinsic conditional expectation $X_t^g\triangleq \eecondg{X_t^{\mathfrak{ext}}}{\gggt}$ must satisfy the SDE
	\begin{align}
	\nonumber
	X_t^{g,i}
	= &
	\lim\limits_{\Delta \mapsto 0^+}
	\langle\EXP{X^g_{t-\Delta}}{\sum_{i=1}^d X_0^i},e_i\rangle_{\rrd} 
	\\
	\nonumber
	+& 
	\int_0^t \left[
	\sum_{i=1}^d 
	\left\langle
	\frac{\partial}{\partial x_i}
	\EXP{X^g_{t-\Delta}}{\sum_{i=1}^d X_t^i},e_i\right\rangle_{\rrd}
	\eecond{f^i(X_u)}{\gggg_u^i}
	\right.
	\\
	\nonumber
	+&\left.
	\frac{1}{2}\sum_{i,j=1}^d 
	\left\langle
	\frac{\partial^2}{\partial x_ix_j}
	\EXP{X^g_{t-\Delta}}{\sum_{i=1}^d X_t^i},e_i
	\right\rangle_{\rrd}
	\Xi_u^{i,j}
	\right]
	du
	\\
	+& 
	\int_0^t
	\sum_{i=1}^d 
	\left\langle
	\frac{\partial}{\partial x_i}
	\EXP{X^g_{t-\Delta}}{\sum_{i=1}^d X_t^i},e_i\right\rangle_{\rrd}
	\eecond{f^i(X_u)}{\gggg_u^i}
	dV_u
	,
	\end{align}
	where the limit is taken with respect to the metric topology on $\LLpgen{\gggt}{2}{\pp}$ and the processes $\Xi_t^{i,j}$ are defined by\small
	$$
	\Xi_t^{i,j} \triangleq \left(
	\eecond{\tilde{X}_u^ic^i}{\gggg_u^i}
	-
	\eecond{\tilde{X}_u^i}{\gggg_u^i}\eecond{c^i(\tilde{X}_u^i)}{\gggg_u^i}
	\right)
	\Big(
	\eecond{X_u^jc^j}{\gggg_u^j}
	-
	\eecond{X_u^j}{\gggg_u^j}\eecond{c^i(X_u^j)}{\gggg_u^j}
	\Big)
	.
	$$\normalsize
\end{cor}
\begin{proof}
	The proof is deferred to the appendix.  
\end{proof}
Corollary~\ref{cor_Main_2} gives a way to use classical Euclidean filtering methods to obtain arbitrarily precise approximations to an SDE for the non-Euclidean conditional expectation.  It is two-fold recursive as it requires the previous non-Euclidean conditional expectation $X_{t-\Delta}^g$ to compute the next update.  In practice, $X_t^g$ will be taken to be the previous asymptotic estimate.  

The next section investigates the numerical performance of the non-Euclidean filtering methodology.  
\section{Numerical Results}\label{ss_numper}
To evaluate the empirical performance of the filtering equations of Corollary~\ref{cor_Main_2}, 1000 successive closing prices ending on September 4$^{th}$ 2018, for the Apple and Google stock are considered.  The unobserved signal process $X_t$ is the covariance matrix between the closing prices at time $t$ and the observation process $Y_t$, is the empirical covariance matrix generated on $7$-day moving windows.

The signal and observation processes $X_t$ and $Y_t$ are assumed to be coupled by equation~\eqref{eq_coupling}.  The functions $f^i$ and $c^i$ are modeled as being deterministic linear functions and $\beta^i, \alpha^i$ are modeled as being constants.  
\begin{equation}
\begin{aligned}
\tilde{X}_t^i =& \int_0^t A^{i,i}\tilde{X}_u^idu
+
\int_0^t C^{i,i}dB_u^i\\
\tilde{Y}_t^i = & \int_0^t H^{i,i}\tilde{X}_u^idu +
\int_0^t K^{i,i}dW_u^i,\\
\tilde{Y}_t^i \triangleq & 
\langle \LOG{X_{t-\frac1{n}}^g}{Y_t},e_i\rangle_{\rrd} \\
\tilde{X}_t^i \triangleq &\langle \LOG{X_{t-\frac1{n}}^g}{X_t},e_i\rangle_{\rrd}\\
B^i\independent W^j,& \, B^i\independent B^j, \, W^i\independent W^j ; i\neq j
\end{aligned}
\label{eq_coupling_Kalman}
\end{equation}
where A, B, C, H, and K are invertible diagonal matrices non-zero determinant.  

Analogous dynamics are for the benchmark methods, ensuring that the Kalman filter is the solution to the stochastic filtering problem.  The values of $A,B,C,H,$ and $K$ are estimated using maximum likelihood estimation.  

Both the classical ($KF$) and proposed methods ($N$-$KF$) are also benchmarked against the non-Euclidean Kalman filtering algorithm of \cite{hauberg2013unscented} ($N$-$KF$-$int$).  This algorithm proposes that the dynamics of $X_t^i$ and $Y_t^i$ be modeled in Euclidean space using the transformations
$$
\begin{aligned}
\tilde{X}_t^i \triangleq \left\langle\LOG{\bar{\Sigma}}{X_t},e_i\right\rangle_{\rrd},\\
\tilde{Y}_t^i \triangleq \left\langle\LOG{\bar{\Sigma}}{Y_t},e_i\right\rangle_{\rrd},\\
\bar{\Sigma} \triangleq \argmin{\Sigma \in \mathscr{P}_D^+}\frac{1}{15}\sum_{j=1}^{15} d_g^2(\Sigma,Y_{t_j}),
\end{aligned}
$$
where $\bar{\Sigma}$ is the intrinsic Riemannian Barycenter (see \cite{bhattacharya2003large} for details properties of the intrinsic mean), and the Riemannian Log and Exp functions are derived from the geometry of $\mathscr{P}_D^+$ and not of $\mmm^{Mrk}$, $15$ was chosen by sequential-validation.  Unlike equations~\eqref{eq_coupling}, the Riemannian log and exp maps are always performed about the same point $\bar{\Sigma}$ and do not update.  This will be reflected in the estimates whose performance progressively degrades over time.  

The Riemannian Barycenter $\bar{\Sigma}$, is computed both intrinsically and extrinsically using the first $15$ empirical covariance matrices.  The extrinsic Riemannian Barycenter on $\mathscr{P}_D^+$ is defined to be the minimizer of
$$
{\bar{\Sigma}^{ext}}\triangleq \EXP{Y_1}{\frac{1}{15}\sum_{j=1}^{15} \LOG{Y_1}{Y_j}}
.
$$
The extrinsic formulation of the Kalman filtering algorithm of \cite{hauberg2013unscented} ($N$-$KF$-$ext$), models the linearized signal and observation processes by
$$
\begin{aligned}
\tilde{X}_t^i \triangleq \left\langle\LOG{{\bar{\Sigma}^{ext}}}{X_t},e_i\right\rangle_{\rrd},\\
\tilde{Y}_t^i \triangleq \left\langle\LOG{{\bar{\Sigma}^{ext}}}{Y_t},e_i\right\rangle_{\rrd}.
\end{aligned}
$$

The length of the moving window was calibrated in a way which maximized the performance of the standard Kalman-filter performed componentwise $(EUC)$.  The choice of $15$ observed covariance matrices used to compute the intrinsic mean was made by sequential validation on the initial $25\%$ of the data.  The findings are reported in the following table.  

\begin{table}[H] \centering 
	\caption{Efficient Portfolio One-Day Ahead Forecasts} 
	\label{tab:MARKOWITZ_BRIEF} 
		\begin{tabular}{@{\extracolsep{5pt}} |r|rr|rr|rr|} 
			\hline
			& $\gamma=0$ & & $\gamma=0.5$ & & $\gamma=1$ & \\
			\hline 
			\hline 
			& $\ell^2$ & $\ell^{\infty}$ & $\ell^2$ & $\ell^{\infty}$ & $\ell^2$ & $\ell^{\infty}$ \\ 
			\hline 
			\textbf{N-KF} & 3.706e-01 & 3.366e-01 & 5.024e-01 & 4.613e-01 & 4.945e-01 & 4.540e-01 \\
			EUC & 6.174e-01 & 5.507e-01 & 7.662e-01 & 6.863e-01 & 7.690e-01 & 6.890e-01 \\ 
			N-KF-int & 5.724e-01 & 5.455e-01 & 7.769e-01 & 7.407e-01 & 7.776e-01 & 7.412e-01 \\ 
			N-KF-ext & 5.244e-01 & 4.804e-01 & 7.078e-01 & 6.515e-01 & 7.051e-01 & 6.497e-01 \\
			\hline 
		\end{tabular} 
\end{table} 

Table~\ref{tab:MARKOWITZ_BRIEF} examines the one day ahead predictive power by evaluating the accuracy of the forecasted portfolio weights.  N-KF is the proposed algorithm.  N-KF-int is the algorithm of \cite{hauberg2013unscented} based on the methods of \cite{fletcher2003statistics}, without the unscented transform. N-KF-int computes the Riemannian $\LOG{\mu}{\cdot}$ and $\LOG{\mu}{\cdot}$ maps where $\mu$ is the intrinsic mean to the first 15 observed covariance matrices and N-KF-ext is the same with the mean computed extrinsically (see \cite{bhattacharya2003large} for a detailed study of intrinsic and extrinsic means on Riemannian manifolds).  %
The one-day ahead predicted weights are evaluated both against the next day's optimal portfolio weights using both the $\ell^{2}$ and $\ell^{\infty}$ norms for portfolios with the risk-aversion levels $\gamma=0,0.5,1$.  

According to each of the performance metrics, the forecasted efficient portfolios using the intrinsic conditional expectation introduced in this paper performs best.  An interpretation is that the Euclidean method disregards all the geometric structure, and that the competing non-Euclidean methods do not update their reference points for the $\EXP{}{}$ and $\LOG{}{}$ transformations.  The failure to update the reference point results in progressively degrading performance.  This effect is not as noticeable when the data is static as in \cite{fletcher2003statistics,fletcher2004principal}, however the time-series nature of the data makes the need to update the reference point for the transformations numerically apparent.  
\begin{table}[ht]
	\centering
	\begin{tabular}{|l|rrrr|}
 \hline
& Frob. & Max Modulus & Inf. & Spectral \\ 
\hline
N-KF & 2.425e-04 & 1.988e-04 & 2.700e-04 & 2.345e-04 \\ 
EUC & 5.043e-04 & 4.041e-04 & 5.525e-04 & 4.772e-04 \\ 
N-KF-int & 4.321e-04 & 3.524e-04 & 4.817e-04 & 4.224e-04 \\ 
N-KF-ext & 5.342e-04 & 4.200e-04 & 6.006e-04 & 5.219e-04 \\ 
\hline
	\end{tabular}
	\caption{Comparison of Covariance Matrix Prediction} 
	\label{tab:Cov_Forecast_comaprison}
\end{table}
Table~\ref{tab:Cov_Forecast_comaprison} examines the covariance matrix forecasts of all four methods directly.  The performance metrics considered are the Frobenius, Maximum Modulus, Infinite and Spectral matrix norms of the difference between the forecasted covariance matrix and the realized future covariance matrix of the two stocks closing prices.  

In Table~\ref{tab:Cov_Forecast_comaprison}, all the non-Euclidean methods out-perform the component-wise classical Euclidean forecasts of the one-day ahead predicted covariance matrix.  The prediction of covariance matrices is less sensitive than that of the efficient portfolio weights, this is most likely due to $\left(\bar{1}^{\star}\Sigma^{-1}\bar{1} \right)^{-1}$ term appearing in equation~\eqref{eq_optim_portfolio} which is sensitive to small changes due to the observably small value of $\Sigma$.  
\begin{table}[H]
	\caption{Bootstrapped Adjusted Confidence Intervals for Performance Metrics}
	\label{tab:Bootstrapped}
	\centering
	\begin{subtable}{.5\textwidth}
		\centering
	\begin{tabular}{|l|rrr|}
 \hline
& 95 L & mean & 95 U \\ 
\hline
Frob. & 4.70e-04 & 5.04e-04 & 5.43e-04 \\ 
Max. Mod. & 3.73e-04 & 4.04e-04 & 4.37e-04 \\ 
Inf. & 5.13e-04 & 5.52e-04 & 5.99e-04 \\ 
Spec. & 4.42e-04 & 4.77e-04 & 5.16e-04 \\ 
\hline
	\end{tabular} 
		\caption{Euclidean Kalman Filter}
	\end{subtable}%
	\begin{subtable}{.5\textwidth}
		\centering
		\begin{tabular}{|l|rrr|}
		    \hline
		  & 95 L & mean & 95 U \\ 
		  \hline
Frob. & 2.02e-04 & 2.42e-04 & 3.04e-04 \\ 
Max. Mod. & 1.61e-04 & 1.98e-04 & 2.53e-04 \\ 
Inf. & 2.25e-04 & 2.70e-04 & 3.36e-04 \\ 
Spec. & 1.97e-04 & 2.34e-04 & 2.92e-04 \\ 
		  \hline
		\end{tabular}
		\caption{Asymptotic Non-Euclidean Kalman Filter}
	\end{subtable}
	\begin{subtable}{.5\textwidth}
		\centering
\begin{tabular}{|l|rrr|}
  \hline
& 95 L & mean & 95 U \\ 
\hline
Frob. & 3.91e-04 & 4.32e-04 & 4.68e-04 \\ 
Max. Mod. & 3.18e-04 & 3.52e-04 & 3.91e-04 \\ 
Inf. & 4.38e-04 & 4.81e-04 & 5.26e-04 \\ 
Spec. & 3.82e-04 & 4.22e-04 & 4.64e-04 \\  
\hline
\end{tabular}
		\caption{Non-Updating Intrinsic Barycenter}
	\end{subtable}%
	\begin{subtable}{.5\textwidth}
		\centering
	\begin{tabular}{|l|rrr|}
		\hline
		& 95 L & mean & 95 U \\ 
		\hline
Frob. & 4.94e-04 & 5.34e-04 & 5.75e-04 \\ 
Max. Mod. & 3.88e-04 & 4.20e-04 & 4.56e-04 \\ 
Inf. & 5.55e-04 & 6.00e-04 & 6.46e-04 \\ 
Spec. & 4.82e-04 & 5.21e-04 & 5.65e-04 \\ 
\hline
	\end{tabular}
		\caption{Non-Updating Extrinsic Barycenter}
	\end{subtable}
\end{table}
Tables~\ref{tab:Cov_Forecast_comaprison} and~\ref{tab:Bootstrapped} report $95\%$ confidence intervals about the estimated mean of the one-day ahead mean error of each respective distance measure.  The error distribution of the performance metrics is non-Gaussian according to the Shapiro-Wilks test performed for normality (see \cite{Shapriowilkstest} for details).  The bootstrap adjusted confidence (BAC) interval method of \cite{diciccio1996bootstrap} is used instead to non-parametrically generate the $95\%$-confidence intervals.  The BAC method is chosen since it does not assume that the underlying distribution is Gaussian, it corrects for bias, and it corrects for skewness in the data.  The bootstrapping was performed by re-sampling $10,000$ times from the realized error distributions of the performance metrics.  %

Tables~\ref{tab:MARKOWITZ_BRIEF} and~\ref{tab:Bootstrapped_Confidence_Intervals_Portfolio_Values} show that the N-KF method is the most accurate and has the lowest variance amongst all the methods according to the Frobenius, Maximum modulus, infinity, and spectral matrix norms.  

\begin{table}[H]
	\caption{Bootstrapped Adjusted Confidence Intervals for Performance Metrics of Portfolio Weights One-Day Ahead Predictions}
	\label{tab:Bootstrapped_Confidence_Intervals_Portfolio_Values}
	\centering
	\begin{subtable}{.5\textwidth}
		\centering
			\begin{tabular}{ |r|l|lll|} 
				\hline 
				\hline
				$\gamma$&	& 95 L & Mean & 95 U \\ 
				\hline
				$0$&		$\ell^2$ & 5.93e-01 & 6.17e-01 & 6.40e-01 \\ 
				&			$\ell^{\infty}$  & 5.32e-01 & 5.51e-01 & 5.70e-01 \\ 
				\hline
				$0.5$&			$\ell^2$ & 7.38e-01 & 7.66e-01 & 7.98e-01 \\ 
				&			$\ell^{\infty}$  & 6.59e-01 & 6.86e-01 & 7.14e-01 \\ 
				\hline
				$1$&			$\ell^2$ & 7.37e-01 & 7.69e-01 & 8.00e-01 \\ 
				&			$\ell^{\infty}$ &  6.61e-01 & 6.89e-01 & 7.15e-01 \\ 
				\hline
			\end{tabular} 
		\caption{Euclidean Kalman Filter}
	\end{subtable}%
	\begin{subtable}{.5\textwidth}
		\centering
			\begin{tabular}{ |r|l|lll|} 
				\hline 
				\hline 
				$\gamma$&			& 95 L & Mean & 95 U \\ 
				\hline 
				$0$&			$\ell^2$ &3.49e-01 & 3.71e-01 & 3.92e-01 \\ 
				&			$\ell^{\infty}$ &  3.18e-01 & 3.37e-01 & 3.55e-01 \\ 
				\hline
				$0.5$&			$\ell^2$ &  4.71e-01 & 5.02e-01 & 5.33e-01 \\ 
				&			$\ell^{\infty}$ & 4.33e-01 & 4.61e-01 & 4.89e-01 \\ 
				\hline
				$1$&			$\ell^2$ &4.66e-01 & 4.94e-01 & 5.25e-01 \\ 
				&			$\ell^{\infty}$ & 4.26e-01 & 4.54e-01 & 4.82e-01 \\ 
				\hline 
			\end{tabular} 
		\caption{Non-Euclidean Kalman Filter}
	\end{subtable}
	\begin{subtable}{.5\textwidth}
		\centering
			\begin{tabular}{ |r|l|lll|} 
				\hline 
				\hline 
				$\gamma$&	& 95 L & Mean & 95 U \\ 
				\hline 
				$0$&			$\ell^2$ & 5.51e-01 & 5.72e-01 & 5.95e-01 \\ 
				&			$\ell^{\infty}$ &5.25e-01 & 5.45e-01 & 5.67e-01 \\ 
				\hline
				$0.5$&			$\ell^2$ & 7.45e-01 & 7.77e-01 & 8.07e-01 \\ 
				&			$\ell^{\infty}$ &7.10e-01 & 7.41e-01 & 7.73e-01 \\ 
				\hline
				$1$&			$\ell^2$ &  7.45e-01 & 7.78e-01 & 8.13e-01 \\ 
				&			$\ell^{\infty}$ & 7.13e-01 & 7.41e-01 & 7.74e-01 \\ 
				\hline
			\end{tabular} 
		
		\caption{Non-Updating Intrinsic Barycenter}
	\end{subtable}%
	\begin{subtable}{.5\textwidth}
		\centering
			\begin{tabular}{ |r|l|lll|} 
				\hline 
				\hline 
				$\gamma$&			& 95 L & Mean & 95 U \\ 
				\hline 
				$0$&			$\ell^2$ & 6.70e-01 & 6.89e-01 & 7.07e-01 \\ 
				&	$\ell^{\infty}$ & 5.97e-01 & 6.14e-01 & 6.31e-01 \\ 
				\hline
				$0.5$&	$\ell^2$ & 8.21e-01 & 8.48e-01 & 8.76e-01 \\ 
				&			$\ell^{\infty}$ & 7.29e-01 & 7.55e-01 & 7.78e-01 \\
				\hline
				$1$&			$\ell^2$ & 8.15e-01 & 8.43e-01 & 8.70e-01 \\ 
				&	$\ell^{\infty}$ &  7.26e-01 & 7.51e-01 & 7.76e-01 \\ 
				\hline
			\end{tabular} 
		\caption{Non-Updating Extrinsic Barycenter}
	\end{subtable}
\end{table}
Tables~\ref{tab:MARKOWITZ_BRIEF} and~\ref{tab:Cov_Forecast_comaprison} reflect that the forecasting performance for the efficient portfolio weights of the N-KF method is more accurate than the others. This is again seen in the lower bias and tighter $95\%$ confidence interval reported in the Table~\ref{tab:Bootstrapped_Confidence_Intervals_Portfolio_Values}.  The numerics presented reflect the importance of incorporating relevant geometry to mathematical finance problems.  The  manner in which non-Euclidean geometry is incorporated in numerical procedures influences the effectiveness of the algorithms as indicated by the superior perfomance of the non-Euclidean Kalman filter over other benchmark methods.

The next section summarizes the contributions made in this paper.  

\section{Summary}\label{sec:summary}
In this paper we have considered non-Euclidean generalizations of conditional expectation which naturally model non-Euclidean features present in probabilistic models.  The need to incorporate relevant geometric information into probabilistic estimation procedures within mathematical finance was motivated by the geometry of efficient portfolios. The connection between geometry and mathematical finance has also been explored in \cite{bjork1999interest},\cite{filipovic2000exponential,filipovic2004geometry},\cite{harms2016consistent}, \cite{henry2005general},\cite{han2015geometric},\cite{bayraktar2006projecting},\cite{brody2004chaos}, and \cite{AFReg,NEU}.   Non-Euclidean filtering was seen to outperform traditional Euclidean filtering methods with the estimates presenting lower prediction errors.

The numerical procedure was justified by Theorem~\ref{thrm_main} which proved the equivalence and existence of common formulations of intrinsic conditional expectation to transformations of a specific Euclidean conditional expectation.  These results were established using the variational-calculus theory of $\Gamma$-convergence introduced in \cite{DeGiorgiGammaConvergence} and subsequently developed by \cite{braides2006handbook} by temporarily passing through the larger $\LLp{\gggg_{\bigcdot}}{p}$-spaces.  To our knowledge, these are novel proofs techniques within the field of mathematical finance and applied probability theory.

A central consequence of Theorem~\ref{thrm_main} is the potential to write down computable stochastic filtering equations for the dynamics of the intrinsic conditional expectation on $(\mmm,g)$ using classical Euclidean filtering equations.  Our results differed from those of \cite{ng1984nonlinear}, \cite{FilteringDuncan}, or \cite{Elworthygeofilt} since dynamics for an intrinsic conditional expectation are forecasted and not dynamics of the Euclidean conditional expectation of a function of a non-Euclidean signal and/or observation process.  Likewise, out results did not rely on the Le Jan-Watanabe connection as those of \cite{LeJanwAtanabeFilteringSalemSIAM} and the only computational bottleneck may be to compute the Riemannian Logarithm and Riemannian Exponential maps.  However, these are readily available in many well-studied geometries not discussed in this paper, for example the hyperbolic geometry used to study the $\lambda$-SABR models in \cite{henry2008analysis}.  

Many other naturally occurring spaces in mathematical finance have the required properties for the central theorems of this paper to apply.  For instance the geometry of two-factor stochastic volatility models developed in \cite{henry2008analysis} do.  The techniques developed here can find applications to that geometry and other relevant geometries in mathematical finance and could find many other areas of applied probability theory where standard machine learning methods have been used extensively.  

\bibliographystyle{abbrvnat}
\bibliography{Mega_Refs}
\printindex

\begin{appendices}	
	\section{Proofs}\label{app:proofs}
	In this, section technical proofs or results from the main body of this section are given.  
	\subsection{Markowitz Space Proof}
	\begin{proof}[Proof of Proposition~\ref{prop_Mark_space_properties}]
		In general, for any three Riemannian manifolds $(M,g^M), (N,g^N)$, $(\tilde{N},g^{\tilde{N}})$, there is a natural bundle-isomorphism $T(M \times N\times \tilde{N}) \cong TM \times TN\times T\tilde{N}$ (see \citep[Section 2.1]{jost2008riemannian} for a discussion on vector bundles).  Under this identification, define the metric on $T(M \times N\times \tilde{N})$ as follows for each $(p,q,r) \in M \times N\times \tilde{N}$.
		$$g^{M\times N\times \tilde{N}}_{(p,q,r)} \colon T_{(p,q,r)}(M \times N\times \tilde{N}) \times T_{(p,q,r)}(M \times N\times \tilde{N}) \to \mathbb{R},$$
		$$((x_1,y_1,z_1),(x_2,y_2,z_3)) \mapsto g^M_p(x_1,x_2,z_3) + g^N_q(y_1,y_2,z_3) + g^{\tilde{N}}_r(y_1,y_2,z_3).$$
		
		Let $\nabla^{M\times N\times \tilde{N}}$ be the Levi-Civita connection on the product of two Riemannian manifolds, then $\nabla^{M\times N} = \nabla^M + \nabla^N + \nabla^{\tilde{N}}$.  Therefore for if $\gamma^M,\gamma^N,\gamma^{\tilde{N}}$ are geodesics on $M$, $N$, $\tilde{N}$ respectively then
		$$
		\nabla^{M\times N\times\tilde{N}}\overset{\cdot}{
			{(\gamma^M,\gamma^N,\gamma^{{\tilde{N}}})}	
		}
		=
		\nabla^M\dot{\gamma}^M + \nabla^N\dot{\gamma}^N + \nabla^{\tilde{N}}\dot{\gamma}^{\tilde{N}}= 0+0+0=0,
		$$
		whence $M\times N\times \tilde{N}$-valued curve $t\mapsto \left(\gamma^M(t),\gamma^N(t)\right)$ is a geodesic on the product Riemannian manifold.  Therefore geodesics, and hence the $\EXP{}{}$ as well as the $\LOG{}{}$ maps can be expressed component-wise on the product Riemannian manifold.  Particularizing $M,N,\tilde{N}$ to $\rr,\rrD$, and $\mathscr{P}_D^+$ implies that the Markowitz space is a well-defined Riemannian manifold.  The formula for $d^{Mrk}$ is just the formula for the product metric between metric spaces.  
		
		Using the natural isomorphism discusses above, the sectional curvature of the product Riemannian is the sum of the sectional curvatures.  Since Euclidean space has $0$-sectional curvature and $\mathscr{P}_D^+$ has non-positive sectional curvature (see\cite{bonnabel2009riemannian}), then the Markowitz space has non-positive sectional curvature.  
		
		The general linear group $GL_D(\rr)$ has two connected components corresponding to the matrices with negative or positive determinant.  Since $\mathscr{P}_D^+$ is a subset comprised of matrices with strictly positive eigenvalues, its elements all have a strictly positive determinant.  Therefore $\mathscr{P}_D^+$ is simply connected.  Since each of the component spaces of the Markowitz space is geodesically complete (see \cite{bonnabel2009riemannian} for the statement concerning $\mathscr{P}_D^+$) the Hopf-Rinow Theorem implies that the associated metric space is complete.  The non-positive curvature of the Markowitz space together with the Cartan-Hadamard Theorem imply that the Riemannian exponential map at every point of the Markowitz space is a diffeomorphism onto the $\rrflex{1+D + \frac{D(D+1)}{2}}$, where $\frac{D(D+1)}{2}$ is the dimension of the $\mathscr{P}_D^+$.  The dimension is obtained by counting the entries on and above the main diagonal of a \textit{symmetric} matrix.  	
	\end{proof}
	\subsection{Derivation of Filtering Equations}
	\begin{proof}[Proof of Corollary~\ref{cor_Main_2}]
		Denote the conditional expectation $\eecond{\tilde{X}_t^i}{\gggt}$ by $X_t^i$.  The filtering equations of \citep[Remark 22.1.15]{CohenAndElliotStochAppl} imply that each of the conditional mean of each locally linearized coordinate processes $\tilde{X}_t^i$ given the filtration $\gggt^i$ is
		\small
		\begin{align}
		X_t^i =& \eecond{X_0^i}{\gggg_0^i} + \int_0^t \eecond{f^i(X_u)}{\gggg_u^i} du%
		+%
		\int_0^u \left(
		\eecond{\tilde{X}_u^ic^i}{\gggg_u^i}
		-
		\eecond{\tilde{X}_u^i}{\gggg_u^i}\eecond{c^i(\tilde{X}_u^i)}{\gggg_u^i}
		\right)
		dV_u
		\label{eq_filtering_Euclideanized_components}
		\end{align}\normalsize
		where the innovations processes $V_t^i$ and the optional projections of $c^i$ are defined by
		$$
		\begin{aligned}
		V_t^i \triangleq & \int_0^t \alpha(u,Y_u)^{-1}dY_u
		- \int_0^t \alpha^i(s,\tilde{Y}_u^i)^{-1}\left(\hat{c}^i(\omega,u,\tilde{Y}_u^i)\right)du\\
		\hat{c}^i(\omega,t,y)\triangleq& \eecond{c^i(t,X_t,y)}{\gggt^i},
		\end{aligned}
		$$
		(see \citep[Chapter 22.10]{CohenAndElliotStochAppl} for more details on the innovations process and \citep[Chapter 7.6]{CohenAndElliotStochAppl} for more details on optional projections). 
		
		Abbreviate $\eecondg{X_t}{\gggt}$ by $Xt^g$
		Applying the It\^{o}-Lemma to the (smooth) function 
		$$
		x\mapsto \left\langle\EXP{X^g_{t-\Delta}}{\sum_{i=1}^d x},e_i\right\rangle_{\rrd}
		$$
		to the process $\sum_{i=1}^dXte_i$ yields
		\begin{align}
		\nonumber
		\langle\EXP{X^g_{t-\Delta}}{\sum_{i=1}^d X_t^i},e_i\rangle_{\rrd} = &
		\langle\EXP{X^g_{t-\Delta}}{\sum_{i=1}^d X_0^i},e_i\rangle_{\rrd} 
		\\
		\nonumber
		+& 
		\int_0^t \sum_{i=1}^d 
		\left\langle
		\frac{\partial}{\partial x_i}
		\EXP{X^g_{t-\Delta}}{\sum_{i=1}^d X_t^i},e_i\right\rangle_{\rrd}
		\,d\hat{X}^i_t 
		\\
		\nonumber
		+& 
		\frac{1}{2}\int_0^t \sum_{i,j=1}^d 
		\frac{\partial^2}{\partial x_ix_j}
		\left\langle
		\EXP{X^g_{t-\Delta}}{\sum_{i=1}^d X_t^i},e_i
		\right\rangle_{\rrd}
		\,d[\hat{X}^i,\hat{X}^j]_t
		\\
		\nonumber
		= &
		\langle\EXP{X^g_{t-\Delta}}{\sum_{i=1}^d X_0^i},e_i\rangle_{\rrd} 
		\\
		\nonumber
		+& 
		\int_0^t \left[
		\sum_{i=1}^d 
		\left\langle
		\frac{\partial}{\partial x_i}
		\EXP{X^g_{t-\Delta}}{\sum_{i=1}^d X_t^i},e_i\right\rangle_{\rrd}
		\eecond{f^i(X_u)}{\gggg_u^i}
		\right.
		\\
		\nonumber
		+&\left.
		\frac{1}{2}\sum_{i,j=1}^d 
		\left\langle
		\frac{\partial^2}{\partial x_ix_j}
		\EXP{X^g_{t-\Delta}}{\sum_{i=1}^d X_t^i},e_i
		\right\rangle_{\rrd}
		\Xi_u^{i,j}
		\right]
		du
		\\
		+& 
		\int_0^t
		\sum_{i=1}^d 
		\left\langle
		\frac{\partial}{\partial x_i}
		\EXP{X^g_{t-\Delta}}{\sum_{i=1}^d X_t^i},e_i\right\rangle_{\rrd}
		\eecond{f^i(X_u)}{\gggg_u^i}
		dV_u
		,
		\end{align}
		where the processes $\Xi_t^{i,j}$ is defined by
		\footnotesize$$
		\Xi_t^{i,j} \triangleq \left(
		\eecond{\tilde{X}_u^ic^i}{\gggg_u^i}
		-
		\eecond{\tilde{X}_u^i}{\gggg_u^i}\eecond{c^i(\tilde{X}_u^i)}{\gggg_u^i}
		\right)
		\left(
		\eecond{X_u^jc^j}{\gggg_u^j}
		-
		\eecond{X_u^j}{\gggg_u^j}\eecond{c^i(X_u^j)}{\gggg_u^j}
		\right)
		.
		$$\normalsize
		The results follow by applying Theorem~\ref{thrm_main} and the Optional Projection \citep[Theorem 7.6.2]{CohenAndElliotStochAppl}.
	\end{proof}

	\subsection{Proof of Theorem~\ref{thrm_main}}
	We return to the proof of Theorem~\ref{thrm_main}.  This will require moving to a slightly larger space where things become more manageable.  %

	\begin{defn}[The $\LLp{\gggg_{\bigcdot}}{p}$ Spaces]\label{defn_Lpbig}
		Let $\tilde{L}_{\pp}^p\left(\gggg_{\bigcdot};\mmm\right)$ denote the subset of the disjoint union $\coprod_{t \in \rr} L_{\pp}^p\left(\gggg_{t\vee 0};\mmm\right)$ consisting of all families $\{X_t\}_{t \in \rr}$ satisfying
		$$
		t \mapsto X_t(\omega) \in D(\rr;\mmm,d_g) ; \pp-a.s.
		$$
		The natural topology on $\tilde{L}_{\pp}^p\left(\gggg_{\bigcdot};\mmm\right)$ induced by these operations will be denoted by $\tau_0$.  
		
		Refine the topology on $\tilde{L}_{\pp}^p\left(\gggg_{\bigcdot};\mmm\right)$ into the coarsest topology on $\tilde{L}_{\pp}^p\left(\gggg_{\bigcdot};\mmm\right)$ satisfying
		\begin{enumerate}[(i)]
			\item $\tau$ is no coarser than the topology on $\tilde{L}_{\pp}^p\left(\gggg_{\bigcdot};\mmm\right)$,
			\item $\{Z_t^n\}_{n \in \nn}$ converges to an element of $\tilde{L}_{\pp}^p\left(\gggg_{\bigcdot};\mmm\right)$ if and only if it converges to $Z_t$ with respect to $\tau$ and $\{Z_{t -\frac1{n}}^n\}_{n \in \nn}$ converges to $Z_t$ in $\tau$.  
		\end{enumerate}
		The one-point compactification of $\tilde{L}_{\pp}^p\left(\gggg_{\bigcdot};\mmm\right)$ is denoted by $\LLp{\gggg_{\bigcdot}}{p}$, the new point, denoted by $\infty$ is called the escape point.  Elements of $\LLp{\gggg_{\bigcdot}}{p}$ are called eternal processes and are denoted by $Z_{\bigcdot}$.  
	\end{defn}
	\begin{remark}\label{ref_rel_top}
		Since $L_{\pp}^p(\gggt;\mmm)$ is a topological subspace of $\LLp{\gggg_{\bigcdot}}{p}$ then it inherits a relative topology.  The indistinguishability discussed in Theorem~\ref{thrm_main} is with respect to this relative topology.  
	\end{remark}
	\begin{remark}[Escape Point]\label{remark_esapce}
		The escape point $\infty$ is interpreted as describing the eternal processes which either fail the finiteness condition of equation~\eqref{eq_defn_LebesgueIntrinsicFinite} or fail to take values in $\mmm$ at a given point in time $\pp$-a.s.  
	\end{remark}
	\begin{remark}[Points in $\LLp{\gggg_{\bigcdot}}{p}$ are Eternal and May Explode]\label{remark_exploding}

		Every element of $\LLp{\gggg_{\bigcdot}}{p}$ is indexed by the time $t$ which takes values in $\rr$ and not only in $[0,\infty)$.  The time $t=0$ is interpreted as when the observer first gained information of the process.  In this way the part above time $t=0$ is a process which may explode arbitrary number of times and the part below is interpreted as a \textit{pre-history} to an observer at time $t=0$.  In this way, processes in $\LLp{\gggg_{\bigcdot}}{p}$ are thought of as eternal.  Note that the eternal process $\exTflex{X}{T}_t$ is $\gggg_{t\wedge T}$-adapted.  
	\end{remark}
	\begin{lem}[Existence]\label{lem_exist}
		The space $\LLp{\gggg_{\bigcdot}}{p}$ exists and is unique up to homeomorphism.  Moreover, $\tilde{L}_{\pp}^p\left(\gggg_{\bigcdot};\mmm\right)$ is dense in $\LLp{\gggg_{\bigcdot}}{p}$.\footnote{These spaces also exhibit universal properties that follow directly from those of the Alexandroff one-point compactification used to construct them, but they are besides the central focus of this section and so will not be discussed here.  }  
	\end{lem}
	\begin{proof}
		The uniqueness and the density of $\tilde{L}_{\pp}^p\left(\gggg_{\bigcdot};\mmm\right)$ are in $\LLp{\gggg_{\bigcdot}}{p}$ the properties of the one-point compactification.  
		
		Let $\tau_0$ denote the topology on $\tilde{L}_{\pp}^p\left(\gggg_{\bigcdot};\mmm\right)$.  Let $\mathscr{T}$ denote the set of topologies containing $\tau_0$ and for which $(ii)$ holds.  $\mathscr{T}$ is non-empty since the discrete topology satisfies both $(i)$ and $(ii)$.  Since the intersections of topologies is again a topology (see \citep[page 55 Problem A.a]{Kelleytop}) then the topology on $\LLp{\gggg_{\bigcdot}}{p}$ exists and is 
		$
		\cap_{\tau \in \mathscr{T}} \tau.
		$	
		Existence follows from the existence of the one point-compactification of the topological space $\left(\tilde{L}_{\pp}^p\left(\gggg_{\bigcdot};\mmm\right),\cap_{\tau \in \mathscr{T}} \tau\right)$.  
	\end{proof}
	The Riemannian Log and Riemannian Exponential maps extend to a correspondence between $\LLp{\gggg_{\bigcdot}}{p}$ and $\mathbb{L}_{\pp}^p\left(\gggg_{\bigcdot};\rrd\right)$.  To see this consider the maps
	\begin{align*}
	\LOGps{}{}:\LLp{\gggg_{\bigcdot}}{p} \times \LLp{\gggg_{\bigcdot}}{p} &\rightarrow \mathbb{L}_{\pp}^p\left(\gggg_{\bigcdot};\rrd\right)
	\\
	\LOGps{Z_{\bigcdot}}{Y_{\bigcdot}} &\mapsto 
	\left\{
	\begin{cases}
	0  &: Z_{\bigcdot} = Y_{\bigcdot} = \infty\\
	\LOG{Z_t}{Y_t}  &: Z_{\bigcdot} \mbox{ and } Y_{\bigcdot} \neq \infty\\
	\infty & :\mbox{else}
	\end{cases}
	\right\}_{t \in \rr}
	,
	\\
	\EXPps{}{}:\LLp{\gggg_{\bigcdot}}{p} \times \mathbb{L}_{\pp}^p\left(\gggg_{\bigcdot};\rrd\right) &\rightarrow \LLp{\gggg_{\bigcdot}}{p}
	\\
	\EXPps{Z_{\bigcdot}}{Y_{\bigcdot}} &\mapsto 
	\left\{
	\begin{cases}
	\EXP{Z_t}{Y_t}  &: Z_{\bigcdot} \neq \infty \mbox{ and } Y_{\bigcdot} \neq \infty\\
	\infty & :\mbox{else}
	\end{cases}
	\right\}_{t \in \rr}
	.
	\end{align*}
	Both these maps collapse to component-wise post-composition by $\LOG{}{}$ (resp. $\EXP{}{})$ if the eternal process $Z_{\bigcdot}$ never hits $\infty$.  
	
	The map $d_g(\cdot,\cdot)$ also induces a map from $\LLp{\gggg_{\bigcdot}}{p}\times \LLp{\gggg_{\bigcdot}}{p}$ into $[0,\infty]$.  The induced map, denoted by ${D_g(\cdot,\cdot)}$ is defined by
	$$
	Z_{\bigcdot} \mapsto \begin{cases}
	d_g(Z_t,X_t) & :\mbox{ if } X_{\bigcdot} \mbox{ and } Z_{\bigcdot} \neq \infty\\
	\infty & : \mbox{else}
	.
	\end{cases}
	$$
	All of these collapse to their usual definitions when the escape point is not encountered.  They will play a key technical role for the remainder of this section.  
	\begin{lem}\label{lem_Gamma_convergence}
		For every $1\leq p<\infty$ and every sub-filtration $\gggg_{\bigcdot}$ of $\fff_{\bigcdot}$, the functionals
		$$
		F_n(Z_{\bigcdot})
		\triangleq
		\int_{t \in \rr}
		\eep{
			\left\|
			\LOGps{Z_{t-\frac1{n}}}{Z_{t}}
			- 
			\LOGps{Z_{t-\frac1{n}}}{X_t}
			\right\|_2^p
		}dt,
		$$
		$\Gamma$-converges to the functional
		$$
		F(Z_{\bigcdot})\triangleq \int_{t \in \rr}
		\eep{
			D_g^p(Z_t,X_t)
		}
		dt
		$$
		on $\LLp{\gggg_{\bigcdot}}{p}$.  
	\end{lem}
	\begin{proof}
		Let $Z_{\bigcdot}$ be an element of $\LLp{\gggt}{p}$, $\{Z_{\bigcdot}^n\}_{n \in \nn}$ be a sequence converging to $Z_{\bigcdot}$ in $\LLp{\gggt}{p}$ and $X_{\bigcdot}$ be an element of $\LLp{\ffft}{p}$.  For every $t\in \rr$, Reverse Fatou's Lemma implies that
		\normalsize	\begin{align}
		\limsup{n \mapsto \infty} & \int_{t \in \rr}\eep{
			\left\|
			\LOGps{Z_{t-\frac1{n}}^n}{Z_{t}^n}
			- 
			\LOGps{Z_{t-\frac1{n}}^n}{X_t}
			\right\|_2^p
		}
		dt\\
		\leq &\int_{t \in \rr}
		\eep{
			\limsup{n \mapsto \infty}
			\left\|
			\LOGps{Z_{t-\frac1{n}}^n}{Z_{t}^n}
			- 
			\LOGps{Z_{t-\frac1{n}}^n}{X_t}
			\right\|_2^p
		}dt
		\label{eq_lae1}
		\end{align}\normalsize
		The continuity of $\|\cdot\|_2^2$, $\LOG{}{}$, and the $\pp$-a.s.\ continuity of the path $t \mapsto Z_{t}(\omega)$ and the choice of topology on $\LLp{\gggg_{\bigcdot}}{p}$ implies that the limit on the RHS of equation~\eqref{eq_lae1} exists and can be computed to be
		\begin{align}
		\nonumber
		\limsup{n \mapsto \infty}& \int_{t \in \rr}\eep{
			\left\|
			\LOGps{Z_{t-\frac1{n}}^n}{Z_{t}^n}
			- 
			\LOGps{Z_{t-\frac1{n}}^n}{X_t}
			\right\|_2^p
		}dt\\
		\leq &
		\int_{t \in \rr}
		\eep{
			\left\|
			\LOGps{Z_{t}}{Z_{t}}
			- 
			\LOGps{Z_{t}}{X_t}
			\right\|_2^p
		}dt\\
		= &
		\nonumber
		\int_{t \in \rr}
		\eep{
			\left\|
			\LOGps{Z_{t}}{X_t}
			\right\|_2^p
		}dt\\
		= &
		\int_{t \in \rr}
		\eep{
			D_g^p(Z_t,X_t)
		}
		dt
		.
		\label{eq_fixing0}
		\end{align}
		Here the fact that $\LOGps{x}{x}=0$ was used along with the relationship between the Riemannian Logarithm and the Riemannian metric, as exemplified in $\mathscr{P}_D^+$ by equation~\eqref{eq_log_PD_pluss}.  
		\normalsize
		Analogously, by the ordinary Fatou's Lemma
		\begin{align}
		\int_{t \in \rr}
		\eep{
			D_g^p(Z_t,X_t)
		}dt
		\leq &
		\liminf{n \mapsto \infty}
		\int_{t \in \rr}
		\eep{
			\left\|
			\LOGps{Z_{t-\frac1{n}}^n}{Z_{t}^n}
			- 
			\LOGps{Z_{t-\frac1{n}}^n}{X_t}
			\right\|_2^p
		}
		dt
		.
		\label{eq_fixing1}
		\end{align}\normalsize
		By the definition of $\Gamma$-convergence, $F$ is the $\Gamma$-limit of the functionals $F_n$ on $\LLp{\gggg_{\bigcdot}}{p}$.  
	\end{proof}
	\begin{ass}\label{ass_finite}
		Both $X_{\bigcdot}^g,X_{\bigcdot}\neq \infty$.  
	\end{ass}
	The proof of Theorem~\ref{thrm_main} relies on a result of central interest in the theory of $\Gamma$-convergence.  This results \citep[Theorem 7.8]{GammaMaso}, is also called the Fundamental Theorem of $\Gamma$-convergence in \citep[Theorem 2.10]{braides2006handbook} in the metric space formulation.  It may be reformulated as stating that if a sequence of functionals $F_n$ $\Gamma$-converges to a functional on a compact topological space\footnote{The assumption of compactness is a special case of the statement which only requires \textit{equicoercivity}.  } $X$, then it must satisfy
	\begin{equation}
	\min_{x \in X} \GamLim{n \mapsto \infty} F_n(x) = \lim\limits_{n \mapsto \infty} \inf_{x \in X} F_n(x)
	.
	\label{fundamental_theorem_I_likeitsosooooomuch}
	\end{equation}
	\begin{proof}[Proof of Theorem~\ref{thrm_main}]
		Lemma~\ref{lem_IC_Comp} established the required $\Gamma$-convergence between the discussed functionals on the compact topological space $\LLp{\gggg_{\bigcdot}}{p}$; this gives existence of the intrinsic conditional expectation $\eecondgp{X_t}{\gggt}{p}$, for every $1\leq p<\infty$.

		For the remainder of this proof, $p$ will be equal to $2$.  Equation ~\eqref{eq_Mark_reg} will be established by an uncountable strong induction, indexed by the totally ordered set $(\rr,\leq)$.  By the definitions of $X_t^g$ and $\eecondgp{%
			X_t
		}{\gggg_0}{p}$%
		if follows that
		$$
		X_0^g = \eecondg{X_0%
		}{\gggg_0} = Z_0.
		$$
		Since $\exTflex{X^g}{0}_t=X_0^g$ and $\eecondg{\exTflex{X}{0}_t%
		}{\gggg_t} = \eecondg{X_0}{\gggg_0}$ for every $t\leq 0$, the base case of the (uncountable) strong induction hypothesis is established.

		Suppose that for every $t\leq T$, $X_t^g = \exTflex{\eecondgp{X_t}{\gggt}{p}}{t}$.  It follows from the $\Gamma$-convergence of $F_n$ to $F$, that
		\footnotesize\begin{align}
		\label{eq_LHRHGamma}
		\min_{Z_{\bigcdot}\in \LLp{\gggg_{\bigcdot}}{p}} 
		\int_{t=0}^{T}
		\eep{
			D_g^p(Z_t,\exTflex{X}{T}_t)
		}
		dt
		= &
		\min_{Z_{\bigcdot}\in \LLp{\gggg_{\bigcdot}}{p}} 
		\int_{t \in \rr}
		\eep{
			D_g^p(Z_t,\exTflex{X}{T}_t)
		}
		dt
		\\
		=&
		\nonumber
		\lim\limits_{n \mapsto \infty} \inf_{Z_{\bigcdot} \in \LLp{\gggg_{\bigcdot}}{p}}
		\int_{t \in \rr}
		\eep{
			\left\|
			\LOGps{Z_{t-\frac1{n}}}{Z_{t}}
			- 
			\LOGps{Z_{t-\frac1{n}}}{\exTflex{X}{T}_t}
			\right\|_2^p
		}dt
		\\
		\nonumber
		=&\lim\limits_{n \mapsto \infty} \inf_{Z_{\bigcdot} \in \LLp{\gggg_{\bigcdot}}{p}}
		\int_{t=0}^{T}
		\eep{
			\left\|
			\LOGps{Z_{t-\frac1{n}}}{Z_{t}}
			- 
			\LOGps{Z_{t-\frac1{n}}}{\exTflex{X}{T}_t}
			\right\|_2^p
		}dt
		.
		\end{align}\normalsize
		Here the fact that $\exTflex{X}{T}_t$ is identical above $T$ and below $0$ was used.  
		The non-negativity of the integrands on of both sides of equation~\eqref{eq_LHRHGamma} and the monotonicity of integration implies that the LHS of equation~\eqref{eq_LHRHGamma} must minimize $\ee_{\pp}\left[D_g^p(Z_t,\exTflex{X}{T}_t)\right]$ for $m$-a.e. value of $t$ between $0$ and the current time $T$.  Therefore by the definition of intrinsic conditional expectation, the left-hand side of equation~\eqref{eq_LHRHGamma} is minimized by the eternal process
		\begin{equation}
		\exTflex{\eecondg{\exTflex{X}{T}_t}{\gggt}}{T}
		\label{eq_explode_optim}
		.
		\end{equation}
		
		Likewise, the right-hand side of equation~\eqref{eq_LHRHGamma} is minimized by the minimizers of
		$$
		\eep{
			\left\|
			\LOGps{Z_{t-\frac1{n}}}{Z_{t}}
			- 
			\LOGps{Z_{t-\frac1{n}}}{\exTflex{X}{T}_t}
			\right\|_2^p
		}.
		$$
		Since $t - \frac1{n}<t$, the induction hypothesis may be, applied hence 
		\begin{equation}
		Z_{t-\frac{1}{n}}=\exTflex{X^g}{t}_{t- \frac1{n}} = \eecondg{\exTflex{X}{t}_{t - \frac1{n}}}{\gggg_{t - \frac1{n}}}
		\label{eq_ind_hyp}
		.
		\end{equation}
		Equation~\eqref{eq_ind_hyp} implies that $\LOGps{\exT{X^g}_{t-\frac1{n}}}{\exT{X}_t}$ no longer enters into the optimization as a variable.  The correspondence between $\LLp{\gggg_{\bigcdot}}{p}$ and $\mathbb{L}_{\pp}^p(\gggg_{\bigcdot};\rrd)$ defined by the map $\LOGps{\exT{X}_{t-\frac1{n}}}{}$ gives
		\small\begin{align}
		\nonumber
		& \inf_{Z_{\bigcdot} \in \LLp{\gggg_{\bigcdot}}{2}}
		\eep{
			\left\|
			\LOGps{\exT{X}_{t-\frac1{n}}}{Z_{t}}
			- 
			\LOGps{\exT{X}_{t-\frac1{n}}}{\exT{X}_t}
			\right\|_2^p
		}\\
		=&
		\inf_{Z_{\bigcdot} \in \mathbb{L}_{\pp}^p(\gggg_{\bigcdot};\rrd)}
		\eep{
			\left\|
			\LOGps{\exT{X}_{t-\frac1{n}}}{Z_{t}}
			- 
			\LOGps{\exT{X}_{t-\frac1{n}}}{\exT{X}_t}
			\right\|_2^p
		}\\
		=& \eecond{\LOGps{\exT{X}_{t-\frac1{n}}}{\exT{X}_t}}{\gggt},
		\end{align}\normalsize
		where the least-squares property of the $L^2$-formulation of conditional expectation (see \citep[page 80]{kallenberg2006foundations}) was used.  Since the Riemannian Logarithm is a diffeomorphism, the change of variables may be undone.  Hence %
		\begin{equation}
		\begin{aligned}
		&\arginf{Z_{\bigcdot} \in \LLp{\gggg_{\bigcdot}}{2}}
		\eep{
			\left\|
			\LOGps{\exTflex{X^g}{t}_{t-\frac1{n}}}{Z_{t}}
			- 
			\LOGps{\exTflex{X^g}{t}_{t-\frac1{n}}}{\exTflex{X}{t}_t}
			\right\|_2^p
		}\\
		=&
		\EXPps{\exTflex{X^g}{t}_{t+\frac1{n}}}{\arginf{\tilde{Z}_{\bigcdot} \in \mathbb{L}_{\pp}^p(\gggg_{\bigcdot};\rrd)}
			\eep{
				\left\|
				\tilde{Z}_t
				- 
				\LOGps{\exTflex{X^g}{t}_{t-\frac1{n}}}{\exTflex{X}{t}_t}
				\right\|_2^p
		}}	
		\\
		=& \EXPps{\exTflex{X^g}{t}_{t+\frac1{n}}}{\arginf{\tilde{Z}_{t} \in L_{\pp}^p(\gggt;\rrd)}
			\eep{
				\left\|
				\tilde{Z}_t
				- 
				\LOGps{\exTflex{X^g}{t}_{t-\frac1{n}}}{\exT{X}{t}_t}
				\right\|_2^p
		}}	
		\\
		=&
		\EXPps{\exTflex{X^g}{t}_{t+\frac1{n}}}{\eecond{\LOGps{\exTflex{X^g}{t}_{t-\frac1{n}}}{\exTflex{X}{t}_t}}{\gggt}}	
		.
		\label{eq_simplif}
		\end{aligned}
		\end{equation}
		\normalsize
		Recombining equations~\eqref{eq_LHRHGamma},~\eqref{eq_explode_optim}, and~\eqref{eq_simplif} yields
		\begin{equation}
		\exTflex{\eecondg{\exTflex{X}{T}_T}{\gggt}}{T}
		=
		\exTflex{\lim\limits_{n \mapsto \infty}%
			\EXPps{\exTflex{X^g}{t}_{T-\frac1{n}}}{\eecond{\LOGps{\exTflex{X^g}{T}_{T-\frac1{n}}}{\exTflex{X}{T}_T}}{\gggt}}	}{T}
		= \exTflex{X^g}{T}_T
		\label{eq_Main_1}
		.
		\end{equation}
		Assumption~\ref{ass_finite} implies that $D_g$, $\LOGps{}{}$, $\EXPps{}{}$ reduce to their usual counterparts.  This completes the induction and 
		establishes Theorem~\ref{thrm_main}.  
	\end{proof}
	The proof of Theorem~\ref{thrm_main} showed how passing through the larger space $\LLp{\gggg_{\cdot}}{2}$ conclusions about the smaller $L^2_{\pp}(\gggt,\mmm)$ spaces could be made.  
\end{appendices}

\end{document}